\documentclass[a4paper,12pt]{article}

\usepackage[scale={.7483,.78}]{geometry}
\usepackage[pdftex]{graphicx}

\usepackage{amsmath,amssymb}
\usepackage{amsthm}
\newtheorem{theorem}{Theorem}
\newtheorem{lemma}[theorem]{Lemma}

\newtheorem*{conjecture}{Conjecture}

\usepackage{CJK,CJKspace,CJKpunct}
\pdfmapline{=unisong@Unicode@ <ipam.ttf}

\newcommand{\figuredir}{./figures}

\usepackage[T1]{fontenc}

\begin{document}

\title{Fence patrolling by mobile agents\\with distinct speeds\footnote{%
A preliminary version of this paper was presented at the
23rd International Symposium on
Algorithms and Computation~\cite{isaac}.}}
\author{%
Akitoshi Kawamura\thanks{%
Supported by Grant-in-Aid for Scientific Research, 
Japan.
}
\and
Yusuke Kobayashi\thanks
{
Supported by Grant-in-Aid for Scientific Research
and by the Global COE Program
``The research and training center for new development in mathematics'', 
MEXT, Japan.
}
}
\date{}

\maketitle

\begin{abstract}
Suppose we want to patrol a fence (line segment) 
using $k$ mobile agents with given speeds $v _1$, \ldots, $v _k$ 
so that every point on the fence 
is visited by an agent at least once in every unit time period. 
Czyzowicz et al.\ conjectured that 
the maximum length of the fence that can be patrolled 
is $(v _1 + \dots + v _k) / 2$, 
which is achieved by the simple strategy 
where each agent~$i$ moves back and forth 
in a segment of length $v _i / 2$. 
We disprove this conjecture by a counterexample involving $k = 6$ agents. 
We also show that the conjecture is true for $k \leq 3$.
\end{abstract}

\section{Introduction}

\emph{Patrolling} is a well-studied task in robotics. 
A set of mobile agents move around a given area 
to protect or supervise it, 
with the goal of ensuring that 
each point in the area is visited frequently enough~%
\cite{Chevaleyre04,Collins13,esa2011,ElmaliachSK2008,YanovskiWB2003}. 
While many authors study heuristic patrolling strategies for various settings
and analyze their performance through experiment, 
recent studies on 
theoretical optimality of strategies 
have revealed that 
there are interesting questions and intricacies 
even in the simplest settings~\cite{esa2011,pasqualetti_franchi_bullo_2010}. 

One of the fundamental problems considered by 
Czyzowicz et al.~%
\cite{esa2011} 
is to patrol a line segment (called the \emph{fence})
using $k$ mobile agents with given speeds. 
They showed that the simple partition-based strategy, 
which is used as parts of many strategies in more general problems~%
\cite{Chevaleyre04, Collins13, ElmaliachSK2008, pasqualetti_franchi_bullo_2010}, 
is optimal in this setting for $k=2$. 
They conjectured that it is also optimal for every $k$. 
In this paper, we prove that the conjecture 
holds for $k = 3$
(Section~\ref{section: optimal}), 
but fails in general (Section~\ref{sec:k=6}). 

\paragraph{Formal description of fence patrolling.}
We are given a line segment of length $l$, which is identified with the interval $[0, l]$. 
A set of points (mobile agents) $a_1, a_2, \dots , a_k$ move along the segment. 
They can move in both directions, and 
can pass one another.
The speed of each agent~$a _i$ may vary during its motion, but
its absolute value is bounded by the 
predefined maximum speed $v _i$. 
The position of agent $a _i$ at time $t$ 
is denoted $a _i (t)$. 
Thus, the motion of the agent $a _i$ is 
described by a function $a _i \colon [0, \infty) \to [0, l]$ 
satisfying $
|a_i(t) - a_i(t+\epsilon)| \leq v_i \cdot \epsilon 
$ for any $t \ge 0$ and $\epsilon >0$. 
A \emph{strategy} (or \emph{schedule}) is given by 
a $k$-tuple of such functions $a _i$. 

For a position $x \in [0, l]$ and time $t ^* \in [0, \infty)$, 
the agent $a _i$ is said to \emph{cover} $(x; t ^*)$ if 
  $a _i(t) = x$ for some $t \in [t^* - 1, t^*)$. 
A strategy is said to 
\emph{patrol} the segment $[0, l]$ 
if 
for any $x \in [0, l]$ and 
  $t ^* \in [1, \infty)$, 
some agent $a _i$ covers $(x; t ^*)$. 

Given the speeds $v_1$, \ldots, $v_k$, 
we want a strategy that patrols the longest possible fence. 
This is equivalent, through scaling, to 
fixing the length of the fence
and minimizing the time, often called the \emph{idle time}, 
during which some point is left unattended by any agent. 

\paragraph{The partition-based strategy.}
An obvious strategy for fence patrolling is as follows: 
partition the fence $[0, l]$ into $k$ segments, 
proportionally to the maximum speeds $v _1$, \ldots, $v _k$, 
and let each agent $a_i$ patrol the $i$th segment
by alternately visiting both endpoints with its maximum speed. 
We call this the \emph{partition-based strategy}. 

Since each agent $a_i$ can patrol a segment of length 
$v _i / 2$, 
the partition-based strategy can patrol a 
segment of length 
$l = (v _1 + \dots + v _k) / 2$. 
Czyzowicz et al.~\cite{esa2011} observed that 
this is optimal when $k=2$. 
They conjectured that it is also the case for every $k$, 
that is, 
a segment of length 
$l > (v _1 + \dots + v _k) / 2$ 
cannot be patrolled. 

In this paper, we disprove this conjecture by demonstrating
$k=6$ agents that patrol
a fence of length greater than 
$(v _1 + \dots + v _k) / 2$
(Theorem~\ref{thm:k=6}). 
On the other hand, we show that
the partition-based strategy is optimal when $k=3$
(Theorem~\ref{thm:k=3}). 


\section{The partition-based strategy is not always optimal}
\label{sec:k=6}

\begin{figure}
\begin{center}
\includegraphics[scale=.77]{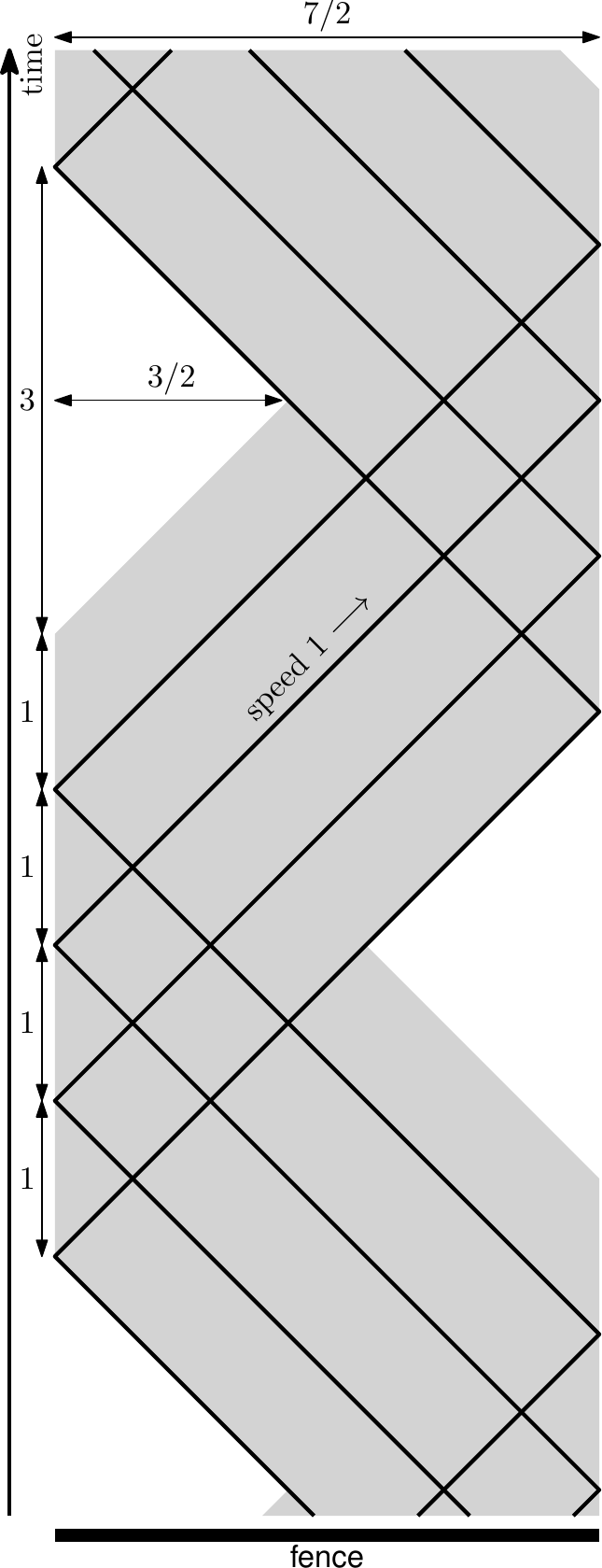}%
\hfill
\includegraphics[scale=.77]{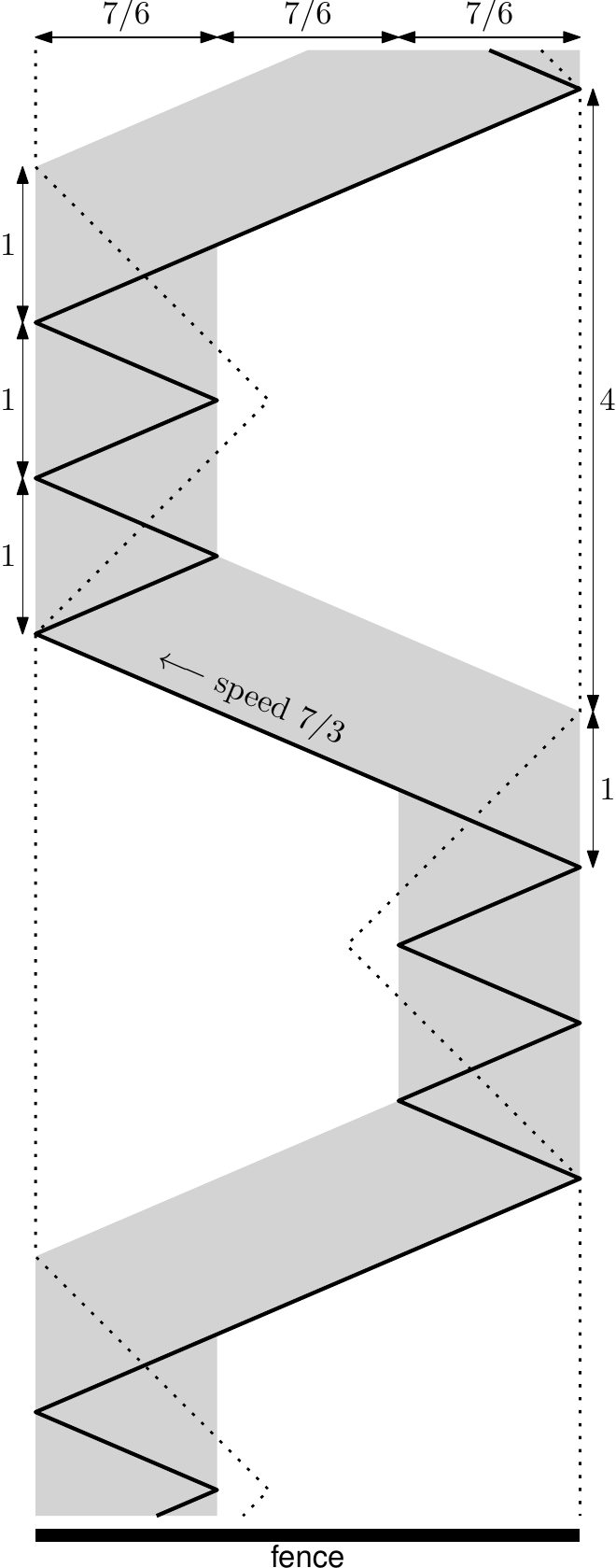}%
\hfill
\includegraphics[scale=.77]{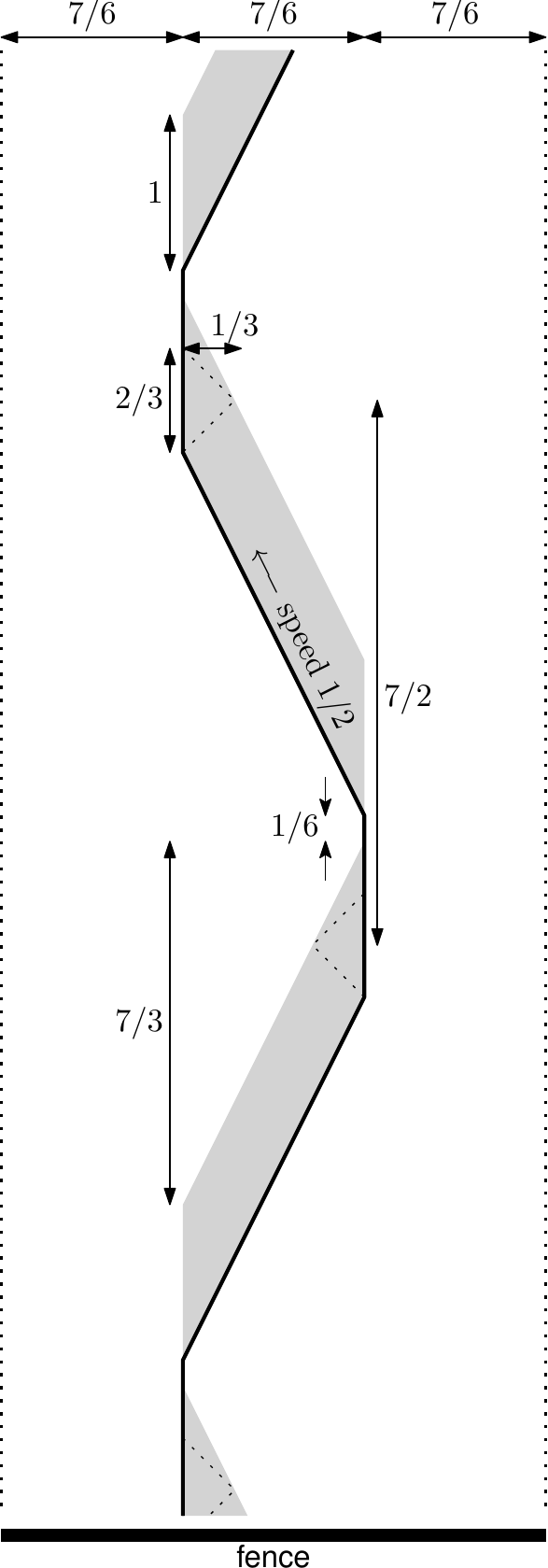}%
\caption{%
Six agents 
patrolling a longer fence than they would with the partition-based strategy. 
}
\label{figure: six agents}
\end{center}
\end{figure}

Fig.~\ref{figure: six agents} shows 
six agents with speeds $1$, $1$, $1$, $1$, $7 / 3$, $1 / 2$
who patrol a fence of length $7 / 2$. 
The fence is placed horizontally and time flows upwards. 
The region covered by each agent is shown shaded (i.e., 
the agent itself moves along the lower edge of each shaded band
of height~$1$). 
This strategy is periodic in the sense that 
each agent repeats its motion every $7$ unit times. 
The four agents with speed~$1$, shown in the diagram on the left, 
visit the two endpoints alternately. 
The region covered by them is shown again by the dotted lines in the middle diagram, 
where another agent with speed $7 / 3$ covers most of the remaining region, but 
misses some small triangles. 
They are covered by the last agent with speed $1 / 2$ 
in the diagram on the right. 
Note that the partition-based strategy with these agents 
would only patrol the length
$(1 + 1 + 1 + 1 + 7 / 3 + 1 / 2) / 2 = 41 / 12 < 7 / 2$. 
Thus, 

\begin{figure}
\begin{center}
\includegraphics[scale=.85]{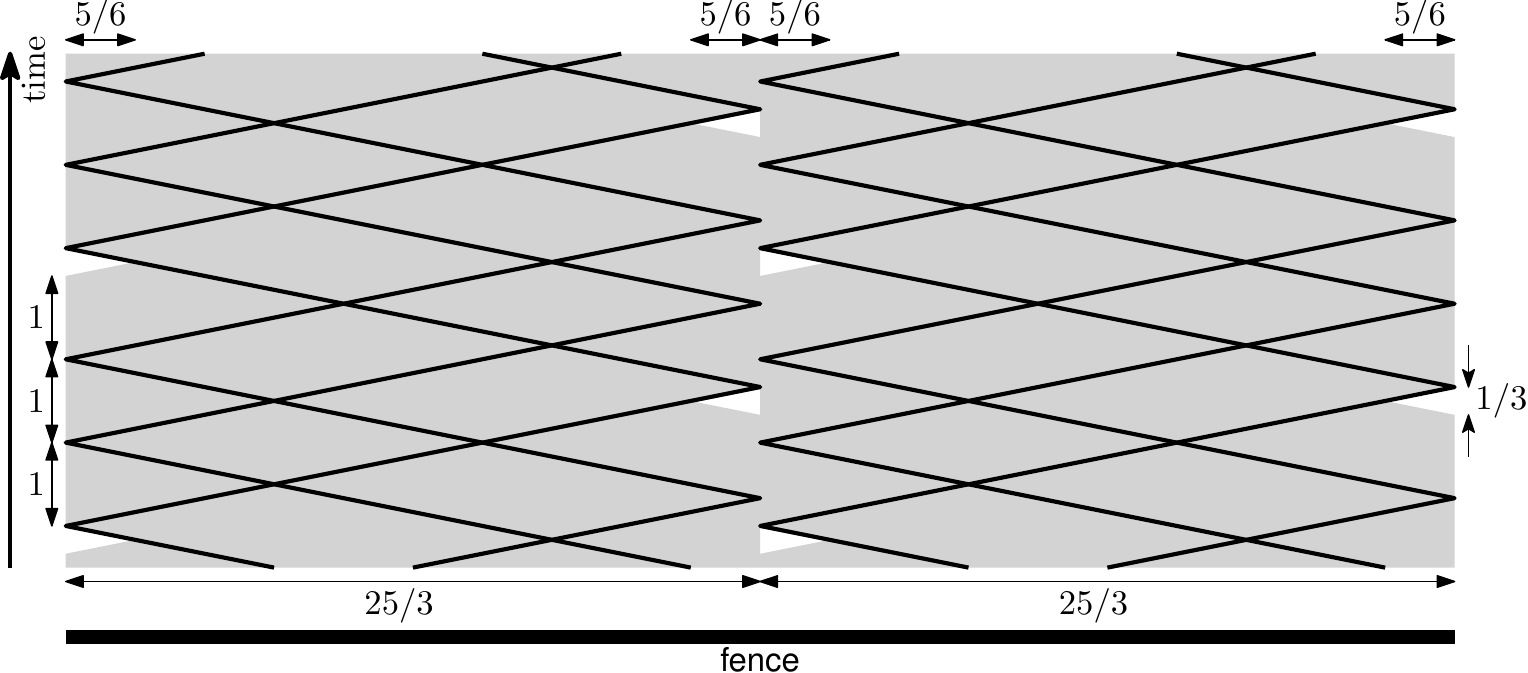}%
\\[10pt]
\includegraphics[scale=.85]{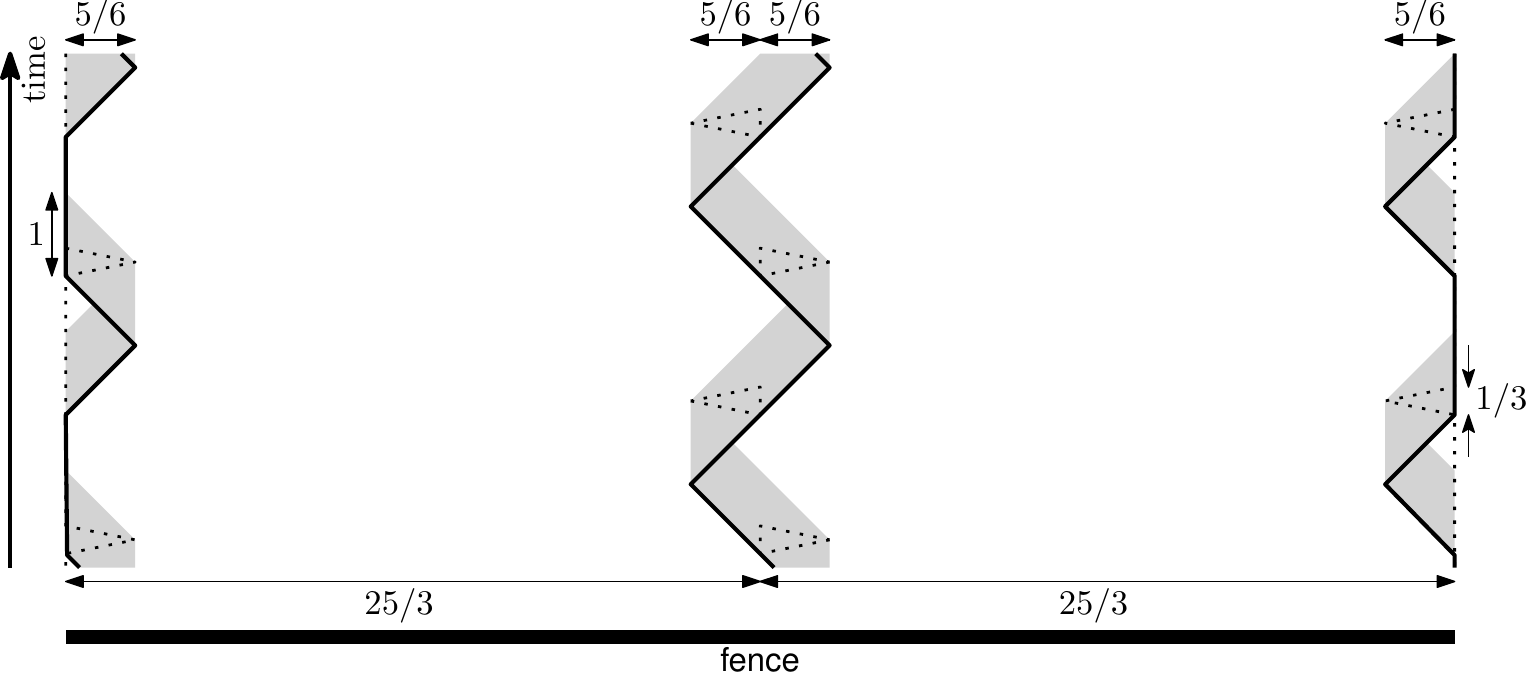}%
\caption{%
Six agents with speed~$5$ (top) and three agents with speed~$1$ (bottom)
together patrolling a longer fence than they would with the partition-based strategy.
}
\label{figure: nine agents}
\end{center}
\end{figure}

\begin{theorem}\label{thm:k=6}
There are settings of agents' speeds for which 
the partition-based strategy is not optimal. 
\end{theorem}

Note that the above example for $k = 6$ agents 
easily implies the non-optimality of the partition-based strategy 
for each $k \geq 6$: 
we can, for example, modify the above strategy 
by extending the fence to the right 
and adding a seventh agent 
who is just fast enough to 
cover the extended part by moving back and forth. 

Another example 
involving more agents but perhaps simpler
is shown in Fig.~\ref{figure: nine agents}, 
where six agents with speed~$5$ and three with speed~$1$ 
patrol a fence of length $50 / 3$ 
using a periodic strategy, with period $10 / 3$.
Here, the six fast agents in the first diagram 
work in two groups of three in a synchronized way. 
The region covered by them is shown again in the second diagram in dotted lines, 
where the missed small triangular regions are covered by the three slow agents. 
The partition-based strategy would only achieve $33 / 2$. 

\section{Cases where the partition-based strategy is optimal}
\label{section: optimal}

Before proving the optimality of 
the partition-based strategy for 
three agents (Section~\ref{sec:k=3}), 
we brief\textcompwordmark ly discuss the much simpler cases of 
equal-speed agents (Section~\ref{section: one speed})
and two agents (Section~\ref{sec:k=2}). 

\subsection{Agents with equal speeds}
\label{section: one speed}

In the homogeneous setting where all agents have the same speed~$v$, 
it is relatively easy to prove that 
the partition-based strategy is optimal. 
This is true more generally when
there are regions that do not have to be visited frequently~%
\cite{Collins13}, 
as well as in related settings where 
the time and locations are discretized in a certain way~%
\cite[Section III]{pasqualetti_franchi_bullo_2010}. 
For the sake of completeness, we provide a short proof for our setting: 

\begin{theorem}
\label{thm:onespeed}
If all agents have the same speed, 
the partition-based strategy is optimal. 
\end{theorem}

\begin{proof}
We proceed by induction on the number $k$ of agents. 
We may assume that 
the agents never switch positions, 
so that $a _1 (t) \leq \dots \leq a _k (t)$ for all $t$. 
This is because 
two agents passing each other could as well just turn back. 
Under this assumption, 
the agent~$a _1$ must visit the point~$0$ once in every unit time, 
and hence is confined to the interval 
$[0, v / 2]$. 
The rest of the fence must be 
patrolled by the other $k - 1$ agents, 
who, by the induction hypothesis, 
cannot do better than the partition-based strategy
which patrols the length $(k - 1) v / 2$. 
Thus the total length is bounded by $v / 2 + (k - 1) v / 2 = k v / 2$. 
\end{proof}

\subsection{Two agents}
\label{sec:k=2}

Although the optimality of the partition-based strategy for two agents 
was already pointed out in \cite{esa2011},  
we present an alternative proof here. 
Some ideas in the proof will be used 
for three agents (Section~\ref{sec:k=3})
and also for the weighted setting (Section~\ref{section: final}). 

\begin{theorem}
\label{thm:k=2}
For two agents, 
the partition-based strategy is optimal. 
\end{theorem}

\begin{proof}
Suppose that 
this was false.  That is, 
suppose that there is a strategy where
agents $a_1$ and $a_2$ 
patrol $[0, l]$ for some
$l > (v _1 + v _2) / 2$. 
We may assume that $v _1 \geq v _2$. 
Let
$l _i = v _i l / (v _1 + v _2)$ for 
$i = 1$, $2$. 
Note that $l = l_1 + l_2$, and 
that it takes time longer than $1 / 2$ 
for agent $a _i$ to travel the distance~$l _i$. 

For any time $t \geq 0$, 
each agent must visit an endpoint ($0$ or $l$)
some time after $t$. 
To see this, let $t _0 > t$ be a time 
at which the endpoint~$0$ is visited. 
Then $(l; t _0 + 1 / 2)$ cannot be covered by this same agent, 
and thus is covered by the other agent. 

Hence, the slower agent $a _2$ 
visits an endpoint, say $0$, at some time $t _2 > 1$.
This implies that $(l _2; t _2 + 1 / 2)$ 
cannot be covered by $a _2$. 
It must therefore be covered by $a _1$, that is, 
$a _1$ must visit $l _2$ at some time $
t _1 \in [t _2 - 1 / 2, t _2 + 1 / 2)
$.  This implies that 
$(l; t_1 + 1 / 2)$ is not covered by $a _1$. 
But it is not covered by $a _2$ either, because $
t _1 + 1 / 2 \in [t _2, t _2 + 1)
$ and the agent $a _2$ cannot travel the distance $l _1 + l _2$ in unit time
(see Fig.~\ref{fig:00}). 
This is a contradiction. 
\end{proof}

\begin{figure}
\begin{center}
\includegraphics[scale=1.0]{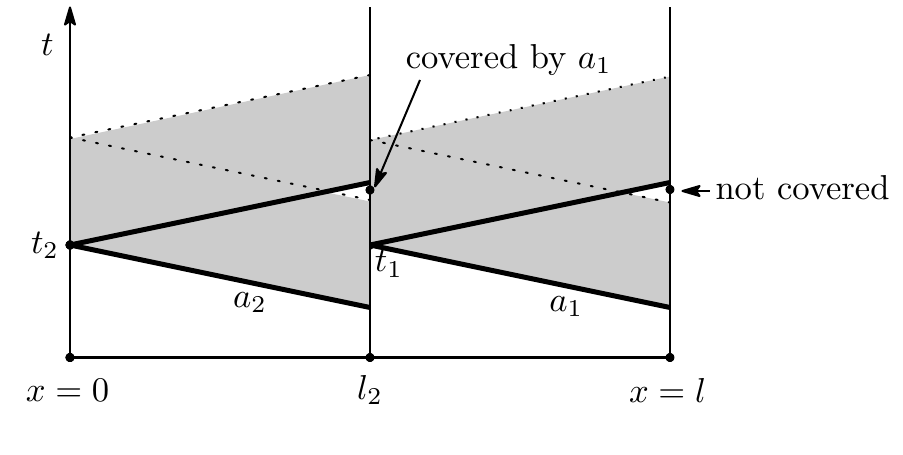}
\caption{Proof of Theorem~\ref{thm:k=2}}
\label{fig:00}
\end{center}
\end{figure}

\subsection{Three agents}
\label{sec:k=3}

In this section, 
we show that 
Czyzowicz et al.'s conjecture is true for three agents: 

\begin{theorem}
\label{thm:k=3}
For three agents, 
the partition-based strategy is optimal. 
\end{theorem}

For a contradiction, suppose that 
agents $a_1$, $a_2$, $a_3$ with speeds $v_1 \geq v_2 \geq v_3$
patrol $[0, l]$, where 
$l > (v_1 + v_2 + v_3) / 2$. 
For $i=1, 2, 3$ let $l_i = v_i l / (v_1 + v_2 + v_3)$, 
so that $l = l_1 + l_2 + l_3$ and 
$l _i > v _i / 2$. 
We start with some lemmas
about the coverage of endpoints. 

\begin{lemma}\label{clm:02}
For any $t^* \geq 0$, 
at least two different agents visit $0$ after the time $t^*$, and 
at least two different agents visit $l$ after the time $t^*$. 
\end{lemma}

\begin{proof}
Let $\{i, j, k\} = \{1, 2, 3\}$, and
assume that $a_i$ is the only agent that visits $0$ after time $t^*$. 
This forces it to stay 
(after time $t ^* + 1 / 2$) 
in the part $[0, l _i]$, 
so the remaining part $[l_i, l]$ of length $l_j + l_k$ has to be patrolled by 
$a_j$ and $a_k$, 
contradicting Theorem~\ref{thm:k=2}. 
The same argument applies to the other endpoint $l$. 
\end{proof}

\begin{lemma}\label{clm:03}
For any $t^* \geq 0$, 
each agent visits at least one of $0$ and $l$ after the time $t^*$. 
\begin{figure}
\begin{center}
\includegraphics[scale=1.0]{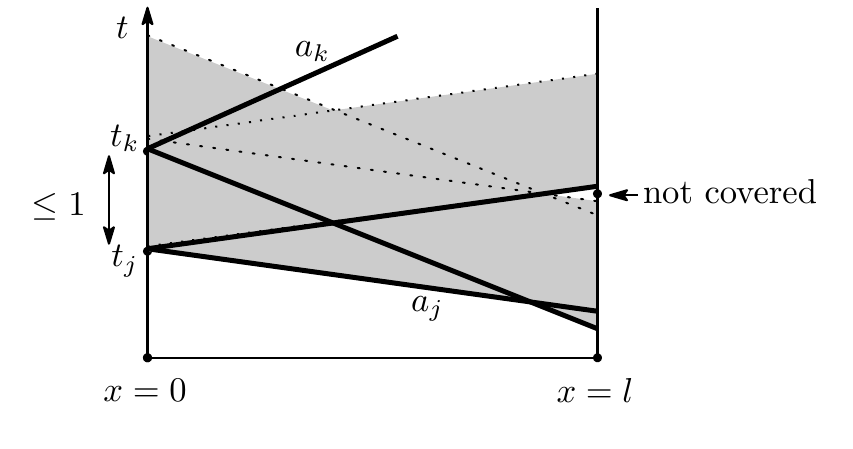}
\caption{Proof of Lemma~\ref{clm:03}}
\label{fig:01}
\end{center}
\end{figure}
\end{lemma}

\begin{proof}
Let $\{i, j, k\} = \{1, 2, 3\}$, and
assume that $a_i$ does not visit $0$ after $t^*$. 
By Lemma~\ref{clm:02}, 
both $a_j$ and $a_k$ visit $0$ infinitely often after $t^*$. 
Thus, 
$a_j(t_j) = a_k(t_k) =0$ for some 
$t_j, t_k > t^* + 1 / 2$ with 
$t_j \leq t_k \leq t_j + 1$ 
(see Fig.~\ref{fig:01}). 
The pair $(l; t_j + l / v _j)$ is not covered by $a_j$, because
  $(t _j + l / v _j) - (t _j - l / v _j) > 1$. 
It is not covered by $a _k$ either, because
\begin{equation*}
 \biggl( t_j + \frac{l}{v_j} \biggr) - \biggl( t_k - \frac{l}{v_k} \biggr) 
>
 t _j + \frac{v _j + v _k}{2 v _j} - t _k + \frac{v _j + v _k}{2 v _k}
=
 (t _j - t _k) + 1 + \frac 1 2 \biggl( \frac{v_k}{v_j} + \frac{v_j}{v_k} \biggr) 
\geq
 1. 
\end{equation*}
Hence, it must be covered by $a_i$, 
which means that $a_i$ visits $l$ after the time $t^*$. 
\end{proof}

\begin{lemma}\label{clm:04}
Suppose that $a_2(t_2) = a_3(t_3) = 0$ (resp.~$= l$) 
for some $t_2, t_3 > 1$. 
Then, 
\begin{itemize}
\item
$a_1(t_1) = 0$ (resp.~$= l$)
for some $t _1 \in (t _2, t _3)$ if $t_2 \leq t_3$, and 
\item
$a_1(t_1) = 0$ (resp.~$= l$)
for some $t _1 \in (t _3, t _2)$ if $t_2 \geq t_3$.
\end{itemize}
\end{lemma}

\begin{proof}
Assume that there are 
$t_3 \geq t_2 > 1$ 
such that 
$a _2 (t _2) = a _3 (t _3) = 0$ 
and $a _1 (t _1) \neq 0$ for any $t _1 \in (t _2, t _3)$. 
We may then retake $t_2$ and $t_3$, if necessary, 
and have $t_3 - t_2 \leq 1$ (see Fig.~\ref{fig:02}). 
\begin{figure}
\begin{center}
\includegraphics[scale=1.0]{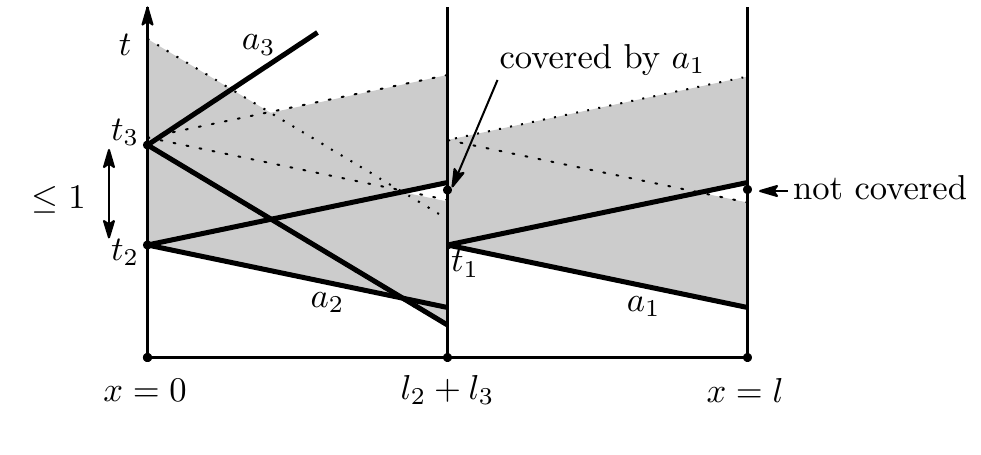}
\caption{Proof of Lemma~\ref{clm:04}}
\label{fig:02}
\end{center}
\end{figure}
By the same argument as the proof of Lemma~\ref{clm:03}, 
the pair $(l_2+l_3; t_2 + (l_2 + l_3) / v_2)$ is covered by neither $a_2$ nor $a_3$. 
More precisely, it is not covered by $a_2$, 
because 
$\bigl( t_2 + (l_2 + l_3) / v_2 \bigr) - \bigl( t_2 - (l_2 + l_3) / v_2 \bigr) > 1$, and
it is not covered by $a_3$ either, 
because 
\begin{equation*}
   \biggl( t_2 + \frac{l_2 + l_3}{v_2} \biggr) 
  -
   \biggl( t_3 - \frac{l_2 + l_3}{v_3} \biggr) 
 > 
  t _2 + \frac{v _2 + v _3}{2 v _2} - t_3 + \frac{v _2 + v _3}{2 v_3}
 =
  (t_2- t_3) + 1 + \frac 1 2 \biggl( \frac{v_3}{v_2} + \frac{v_2}{v_3} \biggr) 
 \geq
  1. 
\end{equation*}
Hence, it must be covered by $a_1$, 
which means that 
$a_1(t_1) = l_2 + l_3$ for some 
$t _1 \in [t _2 + (l _2 + l _3) / v _2 - 1, t _2 + (l _2 + l _3) / v _2)$. 
Since $v_1 \geq v_2 \geq v_3$, $(l; t_1 + l_1 / v_1)$ is covered by none of $a_1, a_2$, and $a_3$, 
which is a contradiction. 

The argument is similar when $t_2 \geq t_3$ and 
when $a_2(t_2) = a_3(t_3) = l$. 
\end{proof}

By Lemmas~\ref{clm:03} and~\ref{clm:04}, 
it happens infinitely often that 
one of the endpoints is visited by $a _1$ and 
then immediately by $a _2$. 
Let us focus on one occurrence of this event, sufficiently later in time 
(time $1 + l / v_3$ is enough), 
which, without loss of generality, happens at the endpoint~$0$. 
That is, 
we fix $t _1$ and $t _2$ with 
$1 + l / v_3 < t_1 \leq t_2 \leq t_1 + 1$ such that 
$a _1 (t _1) = a _2 (t _2) = 0$ and 
no agent visits $0$ during the time interval $(t _1, t _2)$. 
Note that we choose $1 + l / v_3$ so that
every value of time appearing in the proof is at least $1$. 
Now we split into two cases. 

\subsubsection*{Case I: $v_1 \geq 2v_2+v_3$}
\label{sec:k=3case1}

The pair $(l_1 + l_2 - l_3; t_1 + \frac{l_1 + l_2 - l_3}{v_1})$ 
is not covered by $a_1$, 
because $
  \bigl( t_1 + \frac{l_1 + l_2 - l_3}{v_1} \bigr) 
 - 
  \bigl( t_1 - \frac{l_1 + l_2 - l_3}{v_1} \bigr) 
=
 2 \cdot \frac{l_1 + l_2 - l_3}{v_1}
\geq
 \frac{2 l_1}{v_1}
>
 1
$.  It is not covered by $a_2$ either, because 
\begin{align*}
  \biggl( t _1 + \frac{l _1 + l _2 - l _3}{v _1} \biggr) 
 - 
  \biggl( t _2 - \frac{l _1 + l _2 - l _3}{v _2} \biggr) 
&
>
 t _1 + \frac{v _1}{2 v _1} - t _2 + \frac{v _1 + v _2 - v _3}{2 v _2} 
\\
&
=
 (t _1 - t _2) + 1 + \frac{v _1 - v _3}{2 v _2} 
\geq 
 1. 
\end{align*}
Hence, 
it must be covered by $a_3$ (Fig.~\ref{fig:03}), 
\begin{figure}
\begin{center}
\includegraphics[scale=1.0]{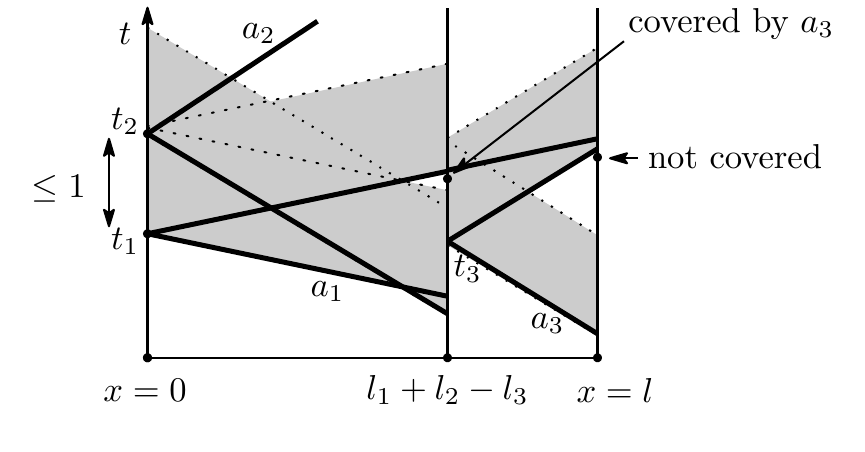}
\caption{Case of $v_1 \geq 2v_2+v_3$}
\label{fig:03}
\end{center}
\end{figure}
which means that $a_3(t_3) = l_1 + l_2 - l_3$
for some $t_3 \in [t_1 + \frac{l_1 + l_2 - l_3}{v_1}-1, t_1 + \frac{l_1 + l_2 - l_3}{v_1})$.  

If $t_3 + \frac{2l_3}{v_3} \leq t_1 + \frac{l}{v_1}$, then 
$(l; t_3 + \frac{2l_3}{v_3})$ is not covered by any of $a_1$, $a_2$, $a_3$
(see Fig.~\ref{fig:03}). 
Otherwise, 
$(l; t_1 + \frac{l}{v_1})$ is not covered by any of $a_1$, $a_2$, $a_3$ 
(not by $a _3$ because $
 (t_1 + \frac{l}{v_1}) - (t_3 - \frac{2l_3}{v_3}) 
>
 t _3 - (t _3 - 1) 
=
 1
$). 

\subsubsection*{Case II: $v_1 \leq 2v_2+v_3$}
\label{sec:k=3case2}

This is the harder case and takes up the rest of this section. 
Again, let $t_1$ and $t_2$ be such that
$1 + \frac{l}{v_3} < t_1 \leq t_2 \leq t_1 + 1$ and 
$a_1(t_1) = a_2(t_2) = 0$. 

\begin{lemma}\label{clm:10}
$a_2(t) \neq l$ for any $t \in [t_1 - \frac{l}{v_1}, t_1 + \frac{l}{v_1}]$. 
\end{lemma}

\begin{proof}
The pair $(l; t_1 + l / v _1)$ is not covered by $a_1$, because
  $(t _1 + l / v _1) - (t _1 - l / v _1) > 1$. 
It is not covered by $a _2$ either, because
\begin{equation*}
 \biggl( t_1 + \frac{l}{v_1} \biggr) - \biggl( t_2 - \frac{l}{v_2} \biggr) 
>
 t _1 + \frac{v _1 + v _2}{2 v _1} - t _2 + \frac{v _1 + v _2}{2 v _2}
=
 (t _1 - t _2) + 1 + \frac 1 2 \biggl( \frac{v_2}{v_1} + \frac{v_1}{v_2} \biggr) 
\geq
 1. 
\end{equation*}
Hence, it must be covered by $a_3$, 
i.e., $a_3$ visits $l$ at some time $
 t' 
\in 
 [t_1 + \frac{l}{v_1} - 1, t_1 + \frac{l}{v_1})
\subseteq 
 [t_1 - \frac{l}{v_1}, t_1 + \frac{l}{v_1}]
$. 
If we assume that 
$a _2(t) = l$ for some time $t \in [t_1 - \frac{l}{v_1}, t_1 + \frac{l}{v_1}]$, 
then, by Lemma~\ref{clm:04},
$a _1 (t'') = l$ for some $t'' \in (t, t') \subseteq (t_1 - \frac{l}{v_1}, t_1 + \frac{l}{v_1})$ (or $t'' \in (t', t) \subseteq (t_1 - \frac{l}{v_1}, t_1 + \frac{l}{v_1})$). 
This contradicts that 
$a_1 (t_1) = 0$ and $|t_1 - t''| < \frac{l}{v_1}$. 
Therefore, we conclude that 
$a _2$ cannot visit $l$ during $[t_1 - \frac{l}{v_1}, t_1 + \frac{l}{v_1}]$. 
\end{proof}

\begin{lemma}\label{clm:05}
$a_3(t) \neq l_1+l_2$  
for any $t \in [t_1 - \frac{l_2 + l_3}{v_1}, t_1 + \frac{l_2 + l_3}{v_1}]$
(see Fig.~\ref{fig:04}). 
\begin{figure}
\begin{center}
\includegraphics[scale=1.0]{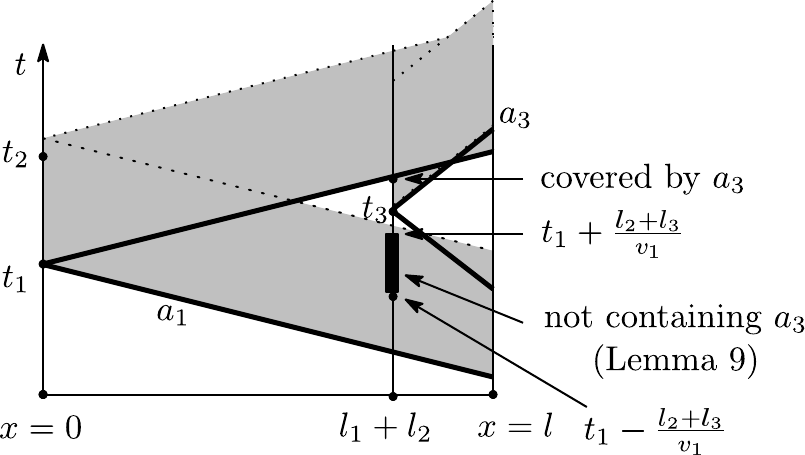}
\caption{Lemmas~\ref{clm:05} and~\ref{clm:06}}
\label{fig:04}
\end{center}
\end{figure}
\end{lemma}

\begin{proof}
Assume that $a_3(t) = l_1+l_2$ for some $t \in [t_1 - \frac{l_2 + l_3}{v_1}, t_1 + \frac{l_2 + l_3}{v_1}]$. 
Then, 
since $[t -\frac{1}{2}, t +\frac{1}{2}] \subseteq [t_1 - \frac{l}{v_1}, t_1 + \frac{l}{v_1}]$, 
neither $a_1$ nor $a_3$ covers $(l; t + \frac{1}{2})$. 
Furthermore, 
by Lemma~\ref{clm:10}, 
$(l; t + \frac{1}{2})$ is not covered by $a_2$ either. 
This is a contradiction. 
\end{proof}

\begin{lemma}\label{clm:06}
$a_3(t) = l_1+l_2$ for some 
$t$ such that $t_1 + \frac{l_2 + l_3}{v_1} < t < t_1 + \frac{l_1 + l_2}{v_1}$ 
(see Fig.~\ref{fig:04}). 
\end{lemma}

\begin{proof}
The pair $(l_1+l_2; t_1 + \frac{l_1 + l_2}{v_1})$ is not covered by $a_1$, 
because 
$\bigl( t_1 + \frac{l_1 + l_2}{v_1} \bigr) - \bigl( t_1 - \frac{l_1 + l_2}{v_1} \bigr) > 1$. 
It is not covered by $a_2$ either, because 
\begin{equation*}
  \biggl( t_1 + \frac{l_1 + l_2}{v_1} \biggr) 
 - 
  \biggl( t_2 - \frac{l_1 + l_2}{v_2} \biggr) 
> 
  t _1 + \frac{v _1 + v _2}{2 v _1} - t _2 + \frac{v _1 + v _2}{2 v _2} 
=
 (t _1- t _2) + 1 + \frac 1 2 \biggl( \frac{v_2}{v_1} + \frac{v_1}{v_2} \biggr) 
\geq
 1. 
\end{equation*}
Hence, it must be covered by $a_3$, 
which means that 
$a_3(t) = l_1 + l_2$ for some $
t \in [t_1 + \frac{l_1 + l_2}{v_1} - 1, t_1 + \frac{l_1 + l_2}{v_1})$. 

Since $
 \frac{l _1 + l _2}{v _1} + \frac{l_2 + l_3}{v_1}
=
 \frac{l _1 + 2 l _2 + l _3}{v _1} 
>
 \frac{v _1 + 2 v _2 + v _3}{2 v _1} 
\geq
 1
$ 
by the assumption $v_1 \leq 2 v_2 + v_3$, 
we have $t_1 + \frac{l_1 + l_2}{v_1} - 1 > t_1 - \frac{l_2 + l_3}{v_1}$. 
Hence, by Lemma~\ref{clm:05}, 
$a_3(t) = l_1 + l_2$ 
for some $t$ such that $t_1 + \frac{l_2 + l_3}{v_1} < t < t_1 + \frac{l_1 + l_2}{v_1}$. 
\end{proof}

Let $t_3$ be the minimum value such that 
$a_3(t_3) = l_1 + l_2$ and $t_1 + \frac{l_2 + l_3}{v_1} < t_3 < t_1 + \frac{l_1 + l_2}{v_1}$
(see Fig.~\ref{fig:04}). 

\begin{lemma}\label{clm:07}
$a_3(t)=l$ for some $t \in [t_1 + \frac{l}{v_1} - 1, t_3 - \frac{l_3}{v_3}]$. 
\end{lemma}

\begin{proof}
The pair $(l; t_1 + \frac{l}{v_1})$ is not covered by $a_1$,
because 
$\bigl( t_1 + \frac{l}{v_1} \bigr) - \bigl( t_1 - \frac{l}{v_1} \bigr) > 1$.
By Lemma~\ref{clm:10}, it is not covered by $a_2$ either. 
Hence, 
$a_3(t) = l$ for some 
$t \in [t_1 + \frac{l}{v_1}-1, t_1 + \frac{l}{v_1})$. 
On the other hand, since $a_3(t_3) = l_1 + l_2$, 
we have $a_3(t) \neq l$ for any $t$ such that $t_3 - \frac{l_3}{v_3} < t < t_3 + \frac{l_3}{v_3}$. 
By combining them, 
we obtain the claim. 
\end{proof}

Let $t'_3$ be the maximum value such that 
$a_3(t'_3)=l$ and 
$t'_3 \in [t_1 + \frac{l}{v _1} - 1, t_3 - \frac{l_3}{v_3}]$. 
Then, 
$(l_1 + l_2; t_3)$ is not covered by $a_3$, 
because 
$t_3 > (t'_3 - \frac{l_3}{v_3}) + 1$ and 
$t'_3 + \frac{l_3}{v_3} > t_1 - \frac{l_2 + l_3}{v_1}$. 
It is not covered by $a_1$ either, because
\begin{equation*}
t_3 > t_1 + \frac{l_2 + l_3}{v_1} > \biggl( t_1 - \frac{l_1 + l_2}{v_1} \biggr) + 1
\end{equation*}
by $v_1 \leq 2v_2 + v_3$. 
Hence, 
it is covered by $a_2$, 
which means that 
$a_2 (t'_2) = l_1 + l_2$ for some $t'_2$ such that 
$t_3 - 1 \leq t'_2 < t_3$ (see Fig.~\ref{fig:05}). 

\begin{figure}
\begin{center}
\includegraphics[scale=1.0]{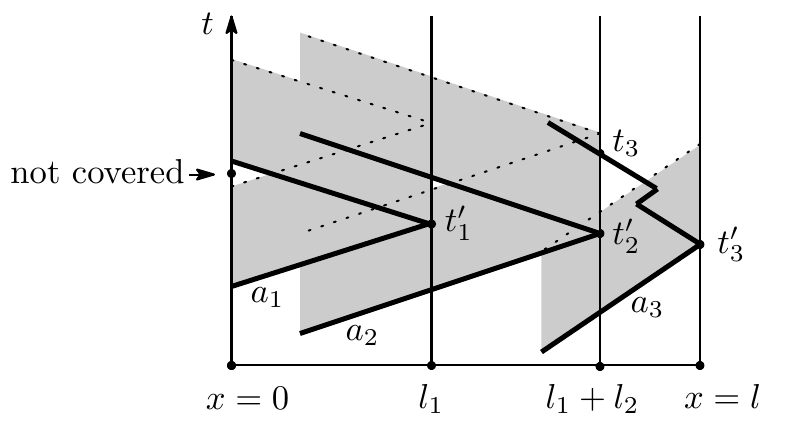}
\caption{Construction of $t'_3$, $t'_2$, and $t'_1$}
\label{fig:05}
\end{center}
\end{figure}

Since $(l_1; t'_2 + \frac{l_2}{v_2})$ is not covered by $a_2$ or $a_3$, 
it is covered by $a_1$, which means that 
$a_1(t'_1) = l_1$ for some $t'_1$ 
such that $t'_2 + \frac{l_2}{v_2} - 1 \leq t'_1 < t'_2 + \frac{l_2}{v_2}$. 
In this case, $(0; t'_1 + \frac{l_1}{v_1})$ is not covered by any of $a_1$, $a_2$, and $a_3$, 
which is a contradiction. 
We have proved Theorem~\ref{thm:k=3}. 

\section{Final remarks}
\label{section: final}

The partition-based strategy is widely used as part of 
multi-agent patrolling strategies. 
We studied its theoretical optimality 
in one of the simplest settings: 
the terrain is a line segment, 
and the agents are points with given maximal speeds. 

\paragraph{The weighted setting.}
It may be natural to consider the \emph{weighted} version of the problem
where each agent has a different power of influence. 
That is, the idle time $T_i > 0$ depends on the agent $a_i$, and
is called the \emph{weight} of $a_i$. 
The setting we have been dealing with in the previous sections
is the special case where $T _i = 1$ for all $i$. 
In the general setting, 
we say that $a_i$ \emph{covers} the pair $(x; t^*)$ if 
$a_i(t) = x$ for some $t \in [t^* -T_i, t^*)$. 
The agents $a _1$, \ldots, $a _k$ are said to \emph{patrol} $[0, l]$ if 
for any $x \in [0, l]$ and $t ^* \in [\max_i ( T_i ), \infty)$, 
the pair $(x; t ^*)$ is covered by some $a _i$. 

As in the unweighted case, 
we can consider the partition-based strategy. 
This time, 
each agent $a_i$ 
is assigned a segment of length proportional to 
the weighted speed $v_i T_i$. 

Theorem~\ref{thm:onespeed} remains true 
in this general setting: 
the partition-based strategy is optimal when the agents have 
different weights $T _i$ but the same speed~$v$. 
To see this, suppose that we could patrol a fence of length 
$l = \alpha + \sum _{i = 1} ^k v T _i / 2$ 
for some $\alpha > 0$. 
Let $\tau = 2 \alpha / k v$. 
Since an agent of weight $T _i$ can be 
simulated by $\lceil T _i / \tau \rceil$ agents of weight~$\tau$
moving in parallel, 
this fence can be patrolled by 
$\kappa = \sum _{i = 1} ^k \lceil T _i / \tau \rceil$ agents, 
all with weight~$\tau$ (and speed~$v$). 
This contradicts (a suitably rescaled version of) Theorem~\ref{thm:onespeed},
since $
 l 
=
 k \tau v / 2 + \sum _{i = 1} ^k T _i v / 2
=
 \sum _{i = 1} ^k (T _i / \tau + 1) v \tau / 2
>
 \kappa v \tau / 2
$. 

Theorem~\ref{thm:k=2} (optimality of the partition-based strategy 
for two agents) also remains true for weighted agents: 
the proof goes through if we set
$l _i = v _i T _i l / (v _1 T _1 + v _2 T _2)$ instead. 

However, Theorem~\ref{thm:k=3} (optimality of the partition-based strategy 
for three agents) fails for the weighted setting. 
To see this, 
consider our first example for Theorem~\ref{thm:k=6} (Fig.~\ref{figure: six agents}), 
and regard the four agents in the left diagram as one agent with weight $4$. 

\paragraph{Summary of our results.}
Thus, our current knowledge can be summarized as follows. 
\begin{itemize}
\item 
 The partition-based strategy is optimal when 
 all agents (possibly weighted) have the same speed 
 (Theorem~\ref{thm:onespeed}), 
 but not when there are two distinct speeds (Fig.~\ref{figure: nine agents}). 
\item 
 The partition-based strategy is optimal when 
 there are two agents with different speeds and weights (Theorem~\ref{thm:k=2}), 
 but not when there are three (Fig.~\ref{figure: six agents}). 
\item 
 The partition-based strategy is optimal when 
 there are three agents with the same weight (Theorem~\ref{thm:k=3}), 
 but not when there are six (Fig.~\ref{figure: six agents}). 
\end{itemize}
The third part settles 
a conjecture of Czyzowicz et al.~\cite{esa2011}, 
but 
our proof for three agents is already quite involved
and seems hard to generalize. 
It remains open whether the partition-based strategy is optimal 
for four and five (unweighted) agents.

\paragraph{Related work and generalizations.}
We considered the patrolling problem
in one of its most basic forms: 
the terrain to be patrolled is a line segment, 
every point in the terrain must be visited, 
and each agent is a point with a maximum speed. 
The problem setting can be generalized in many ways. 
Another simple terrain that has been studied 
in Czyzowicz et al.~\cite{esa2011} is a cycle, 
where again it turns out that 
simple strategies may not be optimal
(see also Dumitrescu, Ghosh and T\'oth~\cite{Dumitrescu2014}). 
Collins et al.~\cite{Collins13} study the patrolling problem
where only part of the fence needs to be visited frequently. 
Chen, Dumitrescu and Ghosh~\cite{Chen13} 
and 
Czyzowicz et al.~\cite{visibility}
discuss agents with some visibility. 
Czyzowicz et al.~\cite{beach} study the setting 
where agents can move faster when walking without watching 
(although their problem is to cover the line segment just once, 
rather than patrolling perpetually). 

For practical purposes, 
it is important to consider decentralized settings 
where agents need to cooperate with limited global knowledge
or computational power~\cite{suzuki-yamashita}. 
The fact that the partition-based strategy is not always optimal
may be bad news in this context, 
since it is one of the simplest strategies
to be realized in a distributed way, 
using systems of self-stabilizing robots, 
e.g., in models of ``bouncing robots''~\cite{bounce}. 
Thus a natural question to ask next is 
whether and how movements better than the partition-based strategy
can be realized in various distributed settings. 

\paragraph{A revised conjecture.}
Since the partition-based strategy 
covers each $(x; t) \in [0, l] \times [1, \infty)$ only doubly, 
it achieves a $2$-approximation
(for the problem of finding the longest possible fence that can be patrolled). 
That is, 
no strategy patrols a fence longer than $v _1 + \dots + v _k$
(in the unweighted setting). 
Although we have shown that 
the partition-based strategy is not always optimal, 
it may still be somewhat close to being optimal, 
given that it is outperformed only slightly 
by our examples for Theorem~\ref{thm:k=6}. 
In other words, the following may be the case, 
with a constant~$c$ fairly close to $1 / 2$: 

\begin{conjecture}
There is a constant $c < 1$ such that 
for any $k$ and any $v _1$, \ldots, $v _k$, 
no strategy can patrol a fence longer than $c (v _1 + \dots + v _k)$. 
\end{conjecture}

The partition-based strategy gives a lower bound of $1 / 2$ 
for such a constant $c$. 
Our first example (Fig.~\ref{figure: six agents}) 
gives $21/41 = 0.5121\ldots {}$. 
After a preliminary version of this paper~\cite{isaac} was presented, 
Chen, Dumitrescu and Ghosh~\cite{Chen13} 
(see also Dumitrescu, Ghosh and T\'oth~\cite{Dumitrescu2014})
improved this bound to $25/48 = 0.5208\ldots {}$. 
Determining the least $c$ is an interesting question. 

\paragraph{Note added for the arXiv version.}
Kawamura and Soejima~\cite{soejima} recently announced 
a lower bound of $2 / 3$. 
The above conjecture still remains open. 

\paragraph{Acknowledgements.}
We thank Yoshio Okamoto for suggesting this research. 
We also thank Taisuke Izumi, Kohei Shimane, Yushi Uno and 
the anonymous referees for helpful comments.

\end{document}